\documentclass[a4paper]{article}
\usepackage{fullpage}
\usepackage{lmodern}
\usepackage[T1]{fontenc}
\usepackage{microtype,xspace}
\usepackage{multirow}
\usepackage{amsfonts,amsmath,amsthm,amssymb,nicefrac,centernot}
\usepackage{graphicx}
\graphicspath{{./figures/}}
\usepackage{subfig}
\usepackage{color}
\usepackage[hidelinks]{hyperref}
\usepackage[parfill]{parskip}

\let\oldproof\proof
\def\proof{\oldproof\unskip}

\makeatletter
\def\thm@space@setup{%
  \thm@preskip=\parskip \thm@postskip=0pt
}
\makeatother

\newtheorem{theorem}{Theorem}
\newtheorem{lem}[theorem]{Lemma}
\newtheorem{corollary}[theorem]{Corollary}

\newtheorem{definition}{Definition}
\newtheorem{conjecture}{Conjecture}

\newcommand\case[1]{\noindent{\descriptionlabel{Case #1.}}\xspace}

\newcommand{\ie}{\emph{i.e.}\xspace}

\newcommand{\R}{\ensuremath{\mathcal{R}}\xspace}
\newcommand{\RR}{\ensuremath{\mathit{\mathcal{RR}}}\xspace}

\newcommand{\ch}{\ensuremath{\mathit{CH}}\xspace}

\newcommand{\seg}[1]{\ensuremath{\overline{#1}}\xspace}
\DeclareMathOperator{\dist}{dist}
\DeclareMathOperator{\length}{length}

\title{Self-approaching paths in simple polygons\thanks{A shorter version of this paper is to be presented at the 33rd International Symposium on Computational Geometry~\cite{sa-socg}.}}

\author{Prosenjit Bose\thanks{School of Computer Science, Carleton University, 
  Ottawa, Canada, 
  \texttt{jit@scs.carleton.ca}}
\and Irina Kostitsyna\thanks{Computer Science Department, Universit\'{e} libre de Bruxelles (ULB), 
  Brussels, Belgium, 
  \texttt{\{irina.kostitsyna,stefan.langerman\}@ulb.ac.be}}
\and Stefan Langerman\footnotemark[3]}

\date{}

\begin{document}

\maketitle

\begin{abstract}

We study \emph{self-approaching paths} that are contained in a simple polygon. A self-approaching path is a directed curve connecting two points such that the Euclidean distance between a point moving along the path and any future position does not increase, that is, for all points $a$, $b$, and $c$ that appear in that order along the curve, $|ac| \ge |bc|$. We analyze the properties, and present a characterization of shortest self-approaching paths. In particular, we show that a shortest self-approaching path connecting two points inside a polygon can be forced to use a general class of non-algebraic curves. While this makes it difficult to design an exact algorithm, we show how to find a self-approaching path inside a polygon connecting two points 
under a model of computation which assumes that we can calculate involute curves of high order.

Lastly, we provide an algorithm to test if a given simple polygon is self-approaching, that is, if there exists a self-approaching path for any two points inside the polygon.
\end{abstract}

\section{Introduction}

The problem of finding an optimal obstacle-avoiding path in a polygonal domain is one of the fundamental problems of computational geometry. Often a desired path has to conform to certain constraints. For example, a path may be required to be monotone~\cite{Arkin1989}, curvature-constrained~\cite{Dubins1957}, have no more than $k$ links~\cite{Mitchell1992c}, etc. A natural requirement to consider is that a point moving along a desired path must always be getting closer to its destination. Such \emph{radially monotone paths} appear, for example, in greedy geographic routing in network setting~\cite{Gao2012} and beacon routing in geometric setting~\cite{Biro2013}. A strengthening of a radially monotone path is a \emph{self-approaching path}~\cite{Icking1995,Icking1999,Aichholzer2001}: a point moving along a self-approaching path is always getting closer not only to its destination but also to all the points on the path ahead of it. There are several reasons to prefer self-approaching paths over radially monotone paths. First, unlike for a radially monotone path, any subpath of a self-approaching path is self-approaching. Therefore, if the destination is not known in advance and the desired path is required to be radially monotone, one will have to resort to using self-approaching paths. Second, the length of a radially monotone path can be arbitrarily large in comparison with the Euclidean distance between the source and the destination points, whereas self-approaching paths have a bounded detour~\cite{Icking1999}.

In this paper we study self-approaching paths that are contained in a simple polygon. We consider the following questions:
\begin{itemize}
\item Given two points $s$ and $t$ inside a simple polygon $P$, does there exist a self-approaching $s$-$t$ path inside $P$?
\item Find the shortest self-approaching $s$-$t$ path.
\item Given a point $s$ in a simple polygon $P$, what is the set of all points reachable from $s$ with self-approaching paths?
\item Given a point $t$, what is the set of all points from which $t$ is reachable with a self-approaching path?
\item Given a polygon $P$, test if it is self-approaching, \ie, if there exists a self-approaching path between any two points in $P$.
\end{itemize}

\subparagraph*{Related work.} Self-approaching curves were first introduced in the context of online searching for a kernel of a polygon~\cite{Icking1995}. They were further studied in~\cite{Icking1999}, where among other results, the authors prove that the length of any self-approaching path connecting two points is not greater than $5.3331$ times the Euclidean distance between the points. An equivalent definition of a self-approaching path is that for every point on the path there has to be a $90^{\circ}$ angle containing the rest of the path. Aichholzer et al.~\cite{Aichholzer2001} developed a generalization of self-approaching paths for an arbitrarily fixed angle $\alpha$ instead of $90^{\circ}$. A relevant type of paths is the increasing chords paths~\cite{Rote1994}, which are self-approaching in both directions. The nice properties of self-approaching and increasing chords paths and their potential to be applied in network routing were recognized by the graph drawing community. As a result, a number of papers appeared in the recent years on self-approaching and increasing chords graphs~\cite{Alamdari2012,Dehkordi2014,Nollenburg2016}.

\subparagraph*{}This paper is organized in the following way. We introduce a few definitions and concepts in Section~\ref{sec:prelim}. In Section~\ref{sec:sa-path-prop}, we characterize a shortest self-approaching path between two points in a simple polygon. In Section~\ref{sec:sa-path} we present an algorithm to construct the shortest self-approaching path between two points if it exists, or to report that it does not exist, by assuming a model of computation in which we can solve certain transcendental equations. 
Finally, in Section~\ref{sec:sa-poly} we present a linear-time algorithm to decide if a polygon is self-approaching, that is, if there is a self-approaching path between any two point of the polygon. In Section~\ref{sec:rev-reach-regs}, we discuss some properties of reachable and reverse-reachable regions.

\section{Preliminaries}\label{sec:prelim}

For two points $p_1$ and $p_2$ on a directed path $\pi$ that starts in point $s$, we shall say that $p_1<_{\pi}p_2$ if $p_1$ lies between $s$ and $p_2$ along $\pi$. For a directed path $\pi$ and two points $p_1<_{\pi}p_2$ on it, denote the sub-path from $p_1$ to $p_2$ by $\pi(p_1,p_2)$.

\begin{definition}
A \emph{self-approaching path} $\pi$ in a continuous domain is a piece-wise smooth\footnote{Some previous works do not require the curve to be smooth. However in this paper we will be mostly considering shortest self-approaching paths, and thus the requirement on smoothness is justified.} oriented curve such that for any three points $a$, $b$, and $c$ on it, such that $a<_\pi b <_\pi c$: $|ac|\ge |bc|$, where $|ac|$ and $|bc|$ are Euclidean distances.
\end{definition}
Icking et al. \cite{Icking1999} showed the following \emph{normal property} of a self-approaching path, that we will be using extensively in this paper,
\begin{lem}[the normal property \cite{Icking1999}]\label{lem:normal}
An $s$-$t$ path $\pi$ is self-approaching if and only if any normal to $\pi$ at any point $a\in\pi$ does not cross $\pi(a,t)$.
\end{lem}
\begin{definition}
A normal $h$ to a directed curve $\pi$ at some point $a\in\pi$ defines two half-planes. Let the \emph{positive half-plane} $h^{+}$ be the open half-plane which is congruent with the direction of $\pi$ at point $a$.
\end{definition}
We can rephrase the normal property in the following way.
\begin{lem}[the half-plane property]\label{lem:half-plane}
An $s$-$t$ path $\pi$ is self-approaching if and only if, for any normal $h$ to $\pi$ at any point $a\in\pi$, the subpath $\pi(a,t)$ lies completely in the positive half-plane $h^{+}$.
\end{lem}

\begin{definition}
A \emph{bend} of a self-approaching path $\pi$ is a point of discontinuity of the first derivative of $\pi$. 
\end{definition}

%

\begin{definition}
A \emph{reachable region} $\R(s)\subseteq P$, for a given point $s$ in a polygon $P$, is the set of all points $t\in P$ for which there exists a self-approaching $s$-$t$ path $\pi \in P$.
\end{definition}

\begin{definition}
A \emph{reverse-reachable region} $\RR(t)\subseteq P$, for a given point $t$ in a polygon $P$, is a set of all points $s\in P$ for which there exists a self-approaching $s$-$t$ path $\pi \in P$.
\end{definition}

\subsection{Involutes}
Next we introduce involute curves of $k$th order that will appear later as parts of shortest self-approaching paths.

\begin{figure}[t]
\centering
\includegraphics{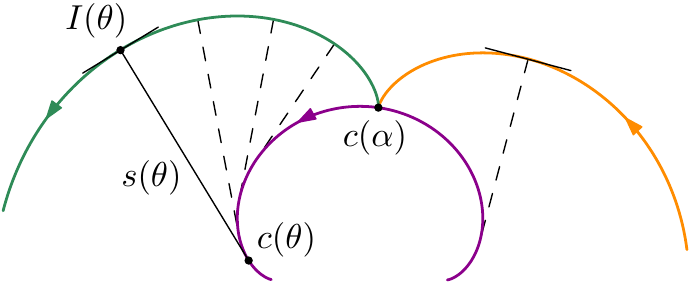}%
\caption{Curve $c(\theta)$ and two involutes. The involute on the left is defined by tangents pointing in the negative direction of $c$, and the involute on the right is defined by tangents pointing in the positive direction of $c$.}
\label{fig:involute-1}%
\end{figure}

An involute of a convex curve $c$ is a curve traced by an end point of an unwinding pull-taut string rolled on $c$. Consider a parameterization $\vec{c}(\theta)$ of the curve, and let $c$ be oriented in the direction of growth of the parameter $\theta$. The involute of $c$ can be computed by the following formula:
\[
\vec{I}(\theta)=\vec{c}(\theta)-s(\theta)\frac{\vec{c}\,'(\theta)}{|\vec{c}\,'(\theta)|}\,,
\]
where $s(\theta)$ is the length of the tangent segment $|\seg{c(\theta)I(\theta)}|$,
\[
s(\theta)=\int\limits_{\alpha}^{\theta} |\vec{c}\,'(t)| dt\,.
\]
The constant $\alpha$ defines the point at which the involute $I$ will start unwinding around $c$ (see Figure~\ref{fig:involute-1}). The involute has two branches: the \emph{positive} branch has the tangent point moving in the positive direction of $c$, and the \emph{negative} branch has the tangent point moving in the negative direction of $c$. If the curve $c$ is defined on the interval $[\theta_{\min},\theta_{max}]$, then the positive branch of its involute is defined on the interval $[\alpha,\theta_{\max}]$, and the negative---on the interval $[\theta_{\min},\alpha]$.

We define an \emph{involute of order $k$} of a curve $c(\theta)$ to be an involute of one branch (that contains the point corresponding to parameter $\alpha_{k}$) of an involute of order $k-1$ of $c(\theta)$, with an involute of order $0$ being the curve $c(\theta)$ itself,
\[
\begin{aligned}
\vec{I}_k(\theta)&=\vec{I}_{k-1}(\theta)-s_{k}(\theta)\frac{\vec{I}'_{k-1}(\theta)}{|\vec{I}'_{k-1}(\theta)|}\,,\qquad\text{where}\qquad s_{k}(\theta)=\int\limits_{\alpha_k}^{\theta} |\vec{I}'_{k-1}(t)| dt\,,\\
\vec{I}_0(\theta)&=\vec{c}(\theta)\,.
\end{aligned}
\]

\begin{figure}[t]
\centering
\includegraphics{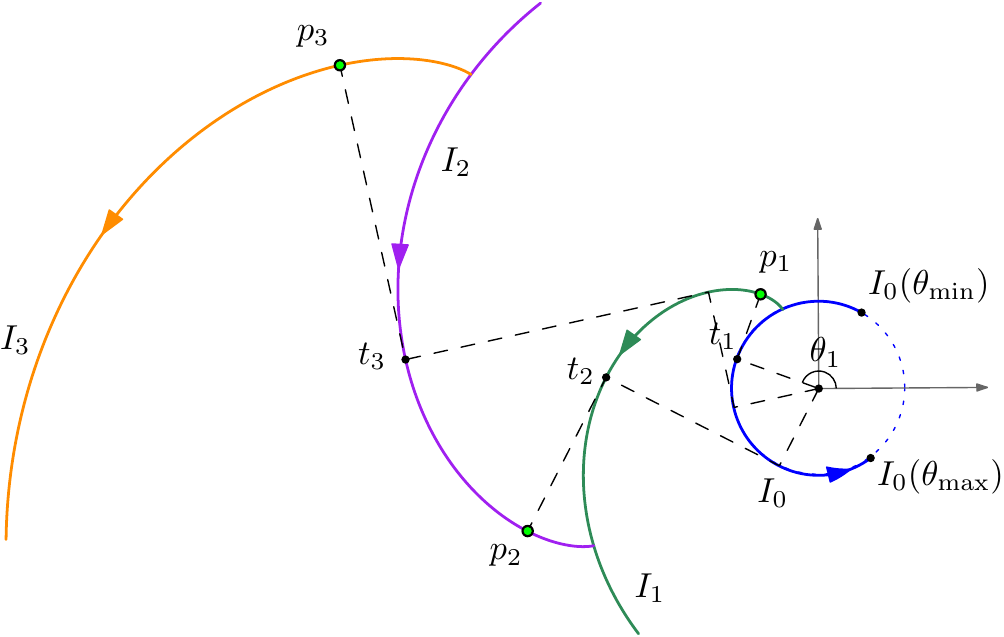}%
\caption{Circular arc $I\!_0(\theta)$, and three involutes $I_1(\theta)$, $I_2(\theta)$, and $I_3(\theta)$: for each $k$, $I\!_k(\theta)$ is an involute about $I\!_{k-1}(\theta)$ that passes through point $p_k$. The arrows designate the direction of growth of the parameter $\theta$.}
\label{fig:involutes}%
\end{figure}

In the following sections we will show that shortest self-approaching paths consist of straight-line segments, circular arcs, and involutes of circular arcs of some order. Next we will derive a formula for an involute of a circle of order $k$.

Consider a circular arc $I\!_0$ that is given by formula
\[
\vec{I}_0(\theta)=r_0\begin{pmatrix}\cos\theta\\ \sin\theta\end{pmatrix}\,,
\]
and is defined for angles in the range $[\theta_{min},\theta_{max}]$.

\paragraph{Involute of order $1$.} Let $I_1$ be the involute of $I\!_0$ that passes through some point $p_1$ with polar coordinates $(r_1,\varphi_1)$.
\[
\vec{I}_1(\theta)=\vec{I}_{0}(\theta)-\int\limits_{\alpha_{0}}^{\theta} |\vec{I}_{0}'(t)| dt\cdot\frac{\vec{I}'_{0}(\theta)}{|\vec{I}'_{0}(\theta)|}
=r_0\begin{pmatrix}\cos\theta\\ \sin\theta\end{pmatrix}-(r_0 \theta+c_1)\begin{pmatrix}-\sin\theta\\ \cos\theta\end{pmatrix}\,.
\]
Denote the parameter value at which $I_1$ passes through $p_1$ as $\theta_1$. Then from the following system of equations
\[
I_1(\theta_1)=r_1\begin{pmatrix}\cos\varphi_1\\ \sin\varphi_1\end{pmatrix}=r_0\begin{pmatrix}\cos\theta_1\\ \sin\theta_1\end{pmatrix}-(r_0 \theta_1+c_1)\begin{pmatrix}-\sin\theta_1\\ \cos\theta_1\end{pmatrix}\,,
\]
it follows that
\[
\begin{aligned}
r_1\cos(\theta_1-\varphi_1)&=r_0\,,\\
r_1\sin(\theta_1-\varphi_1)&=-(r_0\theta_1+c_1)\,.
\end{aligned}
\]
This system has a closed form solution for $\theta_1$ and $c_1$:
\[
\begin{aligned}
\theta_1&=\varphi_1\pm\arccos\frac{r_0}{r_1}\,,\\
c_1&=-\sqrt{r_1^2-r_0^2}-r_0\left(\varphi_1\pm\arccos\frac{r_0}{r_1}\right)\,.
\end{aligned}
\]
Depending on whether the value of $\theta_{1}$ in a given solution falls into the range $[\theta_{\min},\theta_{\max}]$, the involute $I\!_{1}$ can have two branches, one branch, or be undefined.

Let the tangent line drawn from $p_1$ touch $I\!_0$ at point $t_1$. By definition of an involute, point $t_1$ coincides with $I\!_0(\theta_1)$. Then the length of a tangent segment from $p_1$ to $t_1$ is exactly the coefficient at the last term of $I_1(\theta)$ evaluated at parameter $\theta_1$:
\[
|\seg{p_1 t_1}|=-(r_0 \theta_1+c_1)=\sqrt{r_1^2-r_0^2}\,.
\]

\paragraph{Involute of order $2$.} Next, for a selected branch of $I\!_{1}$, we compute an involute of the second order $I_2$ that passes through some point $p_2(r_2,\varphi_2)$:
\[
\begin{split}
I_2(\theta)&=I\!_{1}(\theta)-\int\limits_{\alpha_{1}}^{\theta} |I'_{1}(t)| dt\cdot\frac{I'_{1}(\theta)}{|I'_{1}(\theta)|}\\
&=r_0\begin{pmatrix}\cos\theta\\ \sin\theta\end{pmatrix}-(r_0 \theta+c_1)\begin{pmatrix}-\sin\theta\\ \cos\theta\end{pmatrix}-(r_{0}\frac{\theta^2}{2} + c_1\theta+c_2)\begin{pmatrix}\cos\theta\\ \sin\theta\end{pmatrix}\,,
\end{split}
\]
where
\[
I_2(\theta_2)=r_2\begin{pmatrix}\cos\varphi_2\\ \sin\varphi_2\end{pmatrix}=r_0\begin{pmatrix}\cos\theta_2\\ \sin\theta_2\end{pmatrix}-(r_0 \theta_2+c_1)\begin{pmatrix}-\sin\theta_2\\ \cos\theta_2\end{pmatrix}-(r_{0}\frac{\theta_2^2}{2} + c_1\theta_2+c_2)\begin{pmatrix}\cos\theta_2\\ \sin\theta_2\end{pmatrix}\,,
\]
and subsequently,
\[
\begin{aligned}
r_2\cos(\theta_2-\varphi_2)&=r_0-r_{0}\frac{\theta_2^2}{2} - c_1\theta_2-c_2\,,\\
r_2\sin(\theta_2-\varphi_2)&=-(r_0\theta_2+c_1)\,.
\end{aligned}
\]
These equations can no longer be solved analytically. In this case one has to resort to iterative methods to obtain an approximate solution. Similarly to the previous case, the involute $I\!_{2}$, if defined, can have one or two branches.

Let the tangent line drawn from $p_2$ touch $I_1$ at point $t_2$. Then,
\[
|\seg{p_2 t_2}|=-(r_{0}\frac{\theta_2^2}{2} + c_1\theta_2+c_2)\,.
\]

\paragraph{Involute of order $k$.} Continuing previous calculations in a similar matter we can obtain the following formulas for involutes of order $k=2\ell$ and $k=2\ell+1$:
\[
\begin{aligned}
I\!_{2\ell}(\theta)&=\left(a_0(\theta)-a_2(\theta)+\dots+(-1)^{\ell}a_{2\ell}(\theta)\right)\begin{pmatrix}\cos\theta\\ \sin\theta\end{pmatrix}\\
&-\left(a_1(\theta)-a_3(\theta)+\dots+(-1)^{\ell-1}a_{2\ell-1}(\theta)\right)\begin{pmatrix}-\sin\theta\\ \cos\theta\end{pmatrix}\,,\\
I\!_{2\ell+1}(\theta)&=\left(a_0(\theta)-a_2(\theta)+\dots+(-1)^{\ell}a_{2\ell}(\theta)\right)\begin{pmatrix}\cos\theta\\ \sin\theta\end{pmatrix}\\
&-\left(a_1(\theta)-a_3(\theta)+\dots+(-1)^{\ell+1}a_{2\ell+1}(\theta)\right)\begin{pmatrix}-\sin\theta\\ \cos\theta\end{pmatrix}\,,
\end{aligned}
\]
or shorter,
\[
I\!_{k}(\theta)=\sum_0^{\lfloor\frac{k}{2}\rfloor} (-1)^i a_{2i}(\theta) \begin{pmatrix}\cos\theta\\ \sin\theta\end{pmatrix}-\sum_1^{\lceil\frac{k}{2}\rceil} (-1)^{i-1} a_{2i-1}(\theta)\begin{pmatrix}-\sin\theta\\ \cos\theta\end{pmatrix}\,,
\]
where
\[
\begin{aligned}
a_{i}(\theta)=r_0\frac{\theta^{i}}{i!}+c_1\frac{\theta^{i-1}}{(i-1)!}+\dots+c_i\,.
\end{aligned}
\]
Given a point $p_i(r_i,\varphi_i)$ for each involute $I\!_i$ of order $i$ (for all $1\leq i\leq k$), the constants $c_i$ can be found from the following equations:
\begin{equation}\label{eq:involute}
\begin{aligned}
r_i\cos(\theta_i-\varphi_i)&=a_0(\theta_i)-a_2(\theta_i)+\dots\,,\\
r_i\sin(\theta_i-\varphi_i)&=a_1(\theta_i)-a_3(\theta_i)+\dots\,.
\end{aligned}
\end{equation}
And the length of a tangent segment $\seg{p_k t_k}$ is:
\[
|\seg{p_k t_k}|=|a_k(\theta_k)|\,.
\]

\section{Properties of a shortest self-approaching path}\label{sec:sa-path-prop}

In this section we will prove the following properties of a shortest self-approaching path from $s$ to $t$ inside a simple polygon $P$:
\begin{itemize}
\item A shortest self-approaching path is unique.
\item The shortest self-approaching path consists of straight segments, circular arcs and involutes to the latter pieces of the path.
\end{itemize}
We begin with proving several lemmas:
\begin{lem}
For any two points $p_1<_{\pi}p_2$ on a self-approaching $s$-$t$ path $\pi$ in $\mathbb{R}^2$, the perpendicular bisector of straight-line segment $\seg{p_1p_2}$ does not intersect sub-path $\pi(p_2,t)$.
\end{lem}
\begin{proof}
Let $h^{-}$ be the half-plane defined by the perpendicular bisector of segment $\seg{p_1p_2}$ that contains $p_1$. Assume there is a point $q$ on the subpath $\pi(p_2,t)$ that is interior to $h^{-}$ (refer to Figure~\ref{fig:bisector}). Then $|p_1q|<|p_2q|$, which contradicts the definition of a self-approaching path.
\end{proof}

\begin{figure}[t]
\begin{minipage}[t]{0.48\textwidth}
\centering
\includegraphics{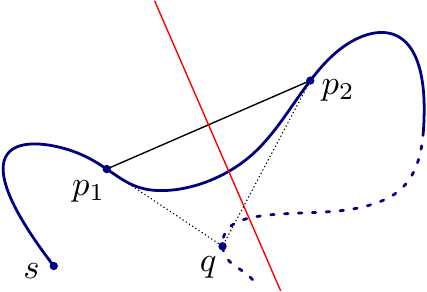}%
\caption{Perpendicular bisector to a segment connecting two points on a path does not intersect the later part of the path.}
\label{fig:bisector}%
\end{minipage}
\hfill
\begin{minipage}[t]{0.48\textwidth}
\centering
\includegraphics{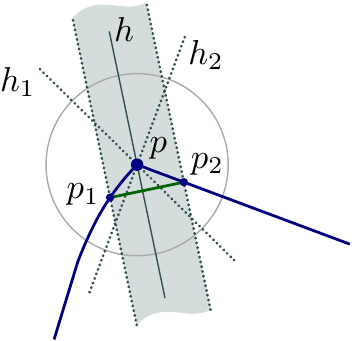}%
\caption{If a self-approaching path has a bend interior to $P$, then there exists a shortcut.}
\label{fig:sa-path-bend}%
\end{minipage}
\end{figure}

\begin{lem}
\label{lem:bends}
Bends of a shortest self-approaching path in a simple polygon $P$ form a subset of vertices of $P$.
\end{lem}
\begin{proof}
Suppose a shortest self-approaching $s$-$t$ path $\pi$ bends at some point $p$ interior to polygon $P$. Then consider an $\varepsilon$-neighborhood of $p$ for some small $\varepsilon$ such that it is interior to $P$, and only contains one connected component of path $\pi$. Let $h_{1}$ and $h_{2}$ be two perpendiculars to $\pi$ at point $p$. Let $h$ be a bisector of an angle formed by $h_{1}$ and $h_{2}$ (as in Figure~\ref{fig:sa-path-bend}). Then, construct a segment $\seg{p_{1}p_{2}}$, perpendicular to $h$, such that point $p_{1}\in\pi(s,p)$, point $p_{2}\in\pi(p,q)$, and two lines parallel to $h$ that pass through $p_{1}$ and $p_{2}$ intersect $h_{1}$ and $h_{2}$ inside the $\varepsilon$-neighborhood of $p$. By the half-plane property, the subpath $\pi(p,t)$ lies completely inside the intersection of two positive half-planes $h_{1}^{+}\cap h_{2}^{+}$. And, because the $\varepsilon$-neighborhood of $p$ contains only one connected component of $\pi$, none of the normal lines to $\seg{p_{1}p_{2}}$ intersects $\pi(p_{2},t)$. Therefore, $\pi(s,p_{1})\oplus \seg{p_{1}p_{2}}\oplus\pi(p_{2},t)$ is self-approaching and is shorter than $\pi$.
\end{proof}

\begin{figure}[t]
\centering
\includegraphics{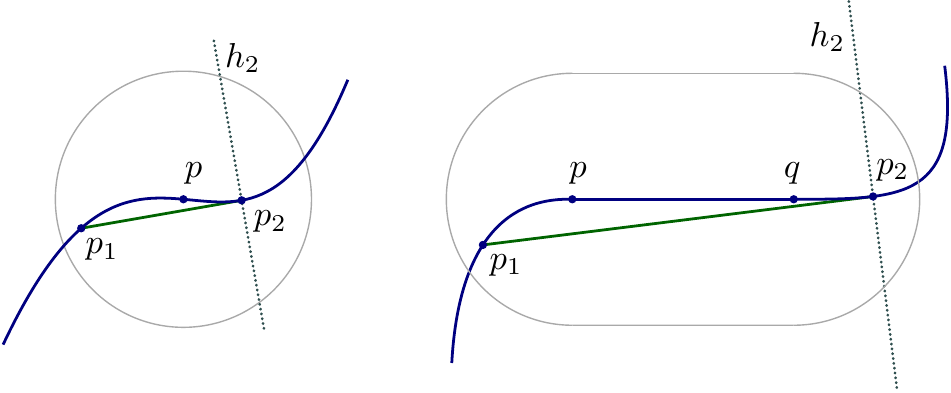}%
\caption{If a self-approaching path has an inflection point (or a segment) interior to $P$, then there exists a shortcut.}
\label{fig:sa-path-infl}%
\end{figure}

Thus, any point of a shortest self-approaching $s$-$t$ path which is interior to $P$ has a well-defined tangent. This point is an \emph{inflection point}, if its tangent separates the self-approaching path in a small enough $\varepsilon$-neighborhood. We can also introduce a notion of an \emph{inflection segment} for a path that contains a straight-line segment as a sub-path. A straight-line segment of a smooth path is an inflection segment if its supporting line separates the path in a small enough $\varepsilon$-neighborhood around the segment (refer to Figure~\ref{fig:sa-path-infl}).

\begin{lem}\label{lem:inflection}
A shortest self-approaching $s$-$t$ path in a simple polygon $P$ cannot have an inflection point (or an inflection segment) that is interior to $P$.
\end{lem}
\begin{proof}
Suppose a shortest self-approaching $s$-$t$ path $\pi$ has an inflection point $p$ (or an inflection segment $\seg{pq}$) interior to polygon $P$. Consider an $\varepsilon$-neighborhood of $p$ (or $\seg{pq}$) for some small $\varepsilon$ such that it is interior to polygon $P$, and it does not contain other inflection points. Choose a point $p_{1}$ on subpath $\pi(s,p)$ close to $p$ and draw a tangent through it to a subpath of $\pi(p,t)$ contained in the $\varepsilon$-neighborhood (refer to Figure~\ref{fig:sa-path-infl}). Let $p_{2}$ be the tangent point. We can always choose point $p_{1}$ to be such that segment $\seg{p_{1}p_{2}}$ lies inside the $\varepsilon$-neighborhood. Let $h_{2}$ be the normal line to $\pi$ drawn through $p_{2}$. Because $\pi$ is self-approaching, the subpath $\pi(p_{2},t)$ lies in the positive half-plane $h_{2}^{+}$. Therefore, none of the normal lines to segment $\seg{p_{1}p_{2}}$ intersects subpath $\pi(p_{2},t)$. Thus, $\pi(s,p_{1})\oplus \seg{p_{1}p_{2}}\oplus\pi(p_{2},t)$ is self-approaching and is shorter than $\pi$.
\end{proof}

Define the \emph{inflection} points of a directed geodesic path $\gamma$ from $s$ to $t$ as the first points of the inflection segments of $\gamma$, \ie, the set of last points in the maximal subchains of $\gamma$ with the same direction of turn. 

\begin{lem}\label{lem:geodesic-inflection-points}
A shortest self-approaching path from $s$ to $t$ in a simple polygon $P$ contains all the inflection points of the geodesic path from $s$ to $t$.
\end{lem}
\begin{proof}
Consider an inflection segment $\seg{p_{i}p_{j}}$ of the geodesic path $\gamma$ from $s$ to $t$, $p_{i}$ is one of its inflection points. Any shortest self-approaching path $\pi$ intersects $\seg{p_{i}p_{j}}$. If the intersection point were not $p_{i}$, then $\pi$ would contain an inflection point that is interior to $P$, but this would contradict Lemma~\ref{lem:inflection}.
\end{proof}

Consider two self-approaching paths $\pi_1$ and $\pi_2$ from $s$ to $t$ in a simple polygon $P$ that do not have other points in common. Let $\gamma$ be a geodesic path from $s$ to $t$ inside the area bounded by $\pi_1$ and $\pi_2$. Then, the following lemma holds.
\begin{lem}\label{lem:geodesic}
A geodesic path $\gamma$ between two self-approaching paths $\pi_1$ and $\pi_2$ is also self-approaching.
\end{lem}
\begin{proof}
\begin{figure}[t]
\centering
\includegraphics{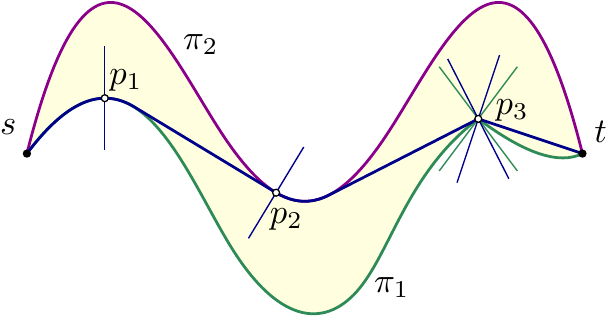}%
\caption{A geodesic bounded between two self-approaching $s$-$t$ paths is also self-approaching.}
\label{fig:geodesic}%
\end{figure}
We use the fact that the geodesic lies inside of the convex hull of each side of the boundaries between which it is constrained, \ie, $\gamma\subset\ch(\pi_1)$ and $\gamma\subset\ch(\pi_2)$.

Any point $p\in\gamma$ either lies on one of the paths $\pi_1$ and $\pi_2$ or on a straight line segment that is bitangent to the boundary (refer to Figure~\ref{fig:geodesic}).

Consider the case when $p$ lies on $\pi_{1}$ or $\pi_{2}$, and is not a bend point (as point $p_1$ in the figure). Let, without loss of generality, $p\in\pi_1$. The positive half-plane $h^{+}$ of the normal to $\pi_1$ at $p$ contains the rest of the path $\pi_1(p,t)$. Therefore it contains the convex hull of $\pi_1(p,t)$, and the subpath $\gamma(p,t)$ of the geodesic.

When $p$ lies on a path $\pi_{1}$ and is a bend point, the two normals to the path at $p$ define two positive half-planes whose intersection contains the rest of the path from $p$ to $t$. The two normals to the geodesic path at this point will lie in between the to normals to the boundary path (as in the figure for point $p_3$). Thus, the intersection of the two positive half-planes of the normals to the geodesic contains the convex hull of the subpath from $s$ to $t$, and, therefore, the rest of the geodesic path $\gamma(p,t)$.

In the case when $p$ lies on a bitangent, consider its end point $p_2$. The normal to $\gamma$ at $p$ is parallel to the normal to $\gamma$ at $p_2$. By one of the cases considered above, the positive half-plane at $p_2$ (or the intersection of two positive half-planes) will contain $\gamma(p_2,t)$, and, therefore, the positive half-plane of normal to $\gamma$ at point $p$ will contain the subpath $\gamma(p,t)$.

Thus, by the half-plane property, $\gamma$ is self-approaching.
\end{proof}
As a corollary to this lemma, for two self-approaching paths from $s$ to $t$, a path, composed of geodesics in the areas bounded by subpaths of the two paths between each pair of consecutive intersection points, is also self-approaching. In other words, let $s=p_0,p_1,\dots,p_k,p_{k+1}=t$ be all the intersection points of $\pi_{1}$ and $\pi_{2}$ in the order they appear on $\pi_{1}$ and $\pi_{2}$. Observe, that the intersection points must appear in the same order along the both paths, otherwise there would exist three points on one of these paths for which the inequality in the definition of a self-approaching path would not be satisfied. Let $\gamma_{i}$ be the geodesic from $p_{i}$ to $p_{i+1}$ in the area between two subpaths $\pi_{1}(p_{i},p_{i+1})$ and $\pi_{2}(p_{i},p_{i+1})$. Then,
\begin{lem}\label{lem:geodesic-2}
The concatenation of the geodesics $\gamma=\gamma_{0}\oplus\gamma_{1}\oplus\cdots\oplus\gamma_{k}$ is self-approaching.
\end{lem}
\begin{proof}
By a similar argument as in Lemma~\ref{lem:geodesic}, for any normal to $\gamma_{i}$ at point $p$, its positive half-plane either contains the convex hull of $\pi_{1}(p,t)$, or it contains the convex hull of $\pi_{2}(p,t)$. In both cases, that implies that the subpath $\gamma(p,t)$ lies in the positive half-plane of the normal. Therefore, $\gamma$ is self-approaching.
\end{proof}
%
%
The next theorem is a direct corollary of Lemma~\ref{lem:geodesic-2}.
\begin{theorem}\label{thm:unique}
A shortest self-approaching $s$-$t$ path is unique.
\end{theorem}
%
%
%
Figure~\ref{fig:shortest-path} shows an example of a shortest self-approaching path inside a polygon. In the next theorem we give its characterization.
\begin{theorem}\label{thm:shortest-path}
A shortest self-approaching $s$-$t$ path in a simple polygon consists of straight segments, circular arcs and circle involutes of some order.
\end{theorem}
\begin{proof}
Let $p_{1},p_{2},\dots,p_{k}$ be the points of the shortest self-approaching $s$-$t$ path $\pi^{*}$ in the order from $s$ to $t$, where it touches the boundary of the polygon $P$. Consider the last segment $\pi^{*}(p_{k},t)$. It is a straight-line segment. Otherwise it could be shortened in the following way. Consider the last segment $\seg{qt}$ of a geodesic path from $s$ to $t$, and extend it in the direction from $t$ to $q$ until intersecting path $\pi^{*}$; denote the intersection point as $q'$. Then, $\pi^{*}$ can be shortened by replacing $\pi^{*}(q',t)$ by the segment $\seg{q't}$.

Now, suppose that all the segments $\pi^{*}(p_{i},p_{i+1})$ consist of straight-line segments, circular arcs, or involutes of a circle of some order for all $i>\ell$ for some $\ell$. We will show, that then, the segment $\pi^{*}(p_{\ell-1},p_{\ell})$ consists of straight-line segments, circular arcs, and/or involutes.

\begin{figure}[t]
\centering
\includegraphics{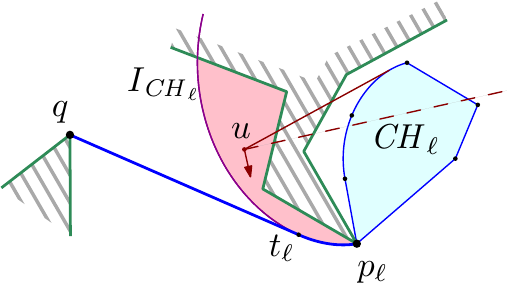}
\caption{Illustration to Theorem~\ref{thm:shortest-path}.}
\label{fig:shortest-path-proof}%
\end{figure}

Denote $\ch\!_{\ell}=\ch(\pi^{*}(p_{\ell},t))$. Let, without loss of generality, $\pi^{*}$ touch the boundary of the polygon at point $p_{\ell}$ on its left side (refer to Figure~\ref{fig:shortest-path-proof}). Then construct an involute $I\!_{\ch\!_{\ell}}$ of the convex hull $\ch\!_{\ell}$ starting at point $p_{\ell}$ with the tangent point moving in the clockwise direction around $\ch\!_{\ell}$ until the first intersection point of the involute with the boundary of the polygon $P$. The area $D_{\ell}$ on the concave side of the involute that it cuts off of the polygon $P$ is a ``dead'' region for any self-approaching path that ends with the subpath $\pi^{*}(p_{\ell,t})$ (red area in the Figure~\ref{fig:shortest-path-proof}). In other words, for any point $u\in D_{\ell}$, any path connecting $u$ to $p_{\ell}$ will have a normal that intersects $\ch\!_{\ell}$, and therefore the subpath $\pi^{*}(p_{\ell},t)$. To show that, consider any piecewise-smooth path $\pi_{u}$ from $u$ to $p_\ell$. Parameterize $\pi_{u}$ for some parameter $\tau\in[0,1]$, where $\pi_{u}(0)=u$ and $\pi_{u}(1)=p_{\ell}$. Consider the distance function $d_{u}(\tau)$ from a point moving along $\pi_{u}$ to the involute $I\!_{\ch\!_{\ell}}$. This function will be piecewise smooth as both of the paths are piecewise-smooth. As a point, moving along $\pi_{u}$, has to eventually coincide with $p_{\ell}$, there exists parameter $\tau'$ at which the distance function is decreasing, and therefore, the angle between a tangent vector to $\pi_{u}$ at the point $u'=\pi_{u}(\tau')$ and a tangent from $u'$ to the convex hull $\ch\!_{\ell}$ is greater than $90^{\circ}$. Therefore, a positive half-plane of the normal to $\pi_{u}$ at point $u'$ does not fully contain the convex hull $\ch\!_{\ell}$, and therefore, the path $\pi_{u}\oplus\pi^{*}(p_{\ell},t)$ is not self-approaching.

Now, consider a geodesic path from $s$ to $p_{\ell}$ in the region $P\backslash D_{\ell}$, and consider its last segment $qp_{\ell}$, where $q$ is the last point before $p_{\ell}$ that belongs to the boundary of $P$. This segment can be a straight-line segment, or a straight-line segment $\seg{qt_{\ell}}$ followed by a piece of the involute $I\!_{\ch\!_{\ell}}$, where $\seg{qt_{\ell}}$ is tangent to $I\!_{\ch\!_{\ell}}$. If segment $qp_{\ell}$ is not on $\pi^{*}$, then, by a similar argument as above, we can show that $\pi^{*}$ can be shortened. Extend the segment $\seg{qt_{\ell}}$ beyond the point $q$ until the intersection $q'$ with $\pi^{*}$. Then, $\pi^{*}$ can be shortened if the subpath $\pi^{*}(q'p_{\ell})$ is replaced by the segment $qp_{\ell}$ of the geodesic.

The boundary of the convex hull $\ch\!_{\ell}$ consists of straight-line segments and pieces of the subpath $\pi^{*}(p_{\ell},t)$, which we assumed were straight segments, arcs, and circle involutes. Therefore, the segment $qp_{\ell}$ of the geodesic path also consists of straight segments, circular arcs, and circle involutes, possibly, of one order higher than the following subpath. Therefore, the shortest self-approaching path consists of straight-line segments, circular arcs, and circle involutes of some order, that is not higher than the number of bends on the path.
\end{proof}

In the last proof, the point $q$ of the last segment $qp_{\ell}$ of the geodesic path from $s$ to $p_{\ell}$ in $P\backslash D_{\ell}$ does not necessarily belong to the geodesic path from $s$ to $t$. Consider an example in Figure~\ref{fig:shortest-path-geodesic}. In it, several vertices of the geodesic path $\gamma$ are in the dead region (on the concave side of the involute). The tangent line from the last vertex ($p_{i}$ in the left example, and $p_{j}$ in the right example) of $\gamma$ before $p_{\ell}$ that is not in the dead region intersects the boundary of the polygon. Angle $\angle p_{i}gt_{\ell}$, where $g$ is the intersection point of $\gamma$ with the involute, is an obtuse angle. This follows from the fact that the intersection angle between the straight-line segment $\seg{p_{i}p_{\ell}}$ and the tangent to the involute at the intersection point must not be greater than $90^{\circ}$, otherwise the point $p_{\ell}$ would not lie in the positive half-plane of the normal to the involute at the intersection point. Then, the total turn angle of the self-approaching path $\pi^{*}$ from $p_{i}$ to $t_{\ell}$ is less than $90^{\circ}$, and thus, the subpath $\pi^{*}(p_{i},t_{\ell})$ consists of straight-line segments. Let the previous inflection point of $\gamma$ before $p_{\ell}$ be $p_{j}$, and the next inflection point of $\gamma$ on or after $p_{\ell}$ be $p_{k}$. It follows then that the subpath $\pi^{*}(p_{j},p_{k})$ is \emph{geodesically convex}, that is, the shortest path between any two points on $\pi^{*}(p_{j},p_{k})$ lies completely on one (and the same side) of the path. We obtain the following lemma.

\begin{figure}[t]
\centering
\includegraphics{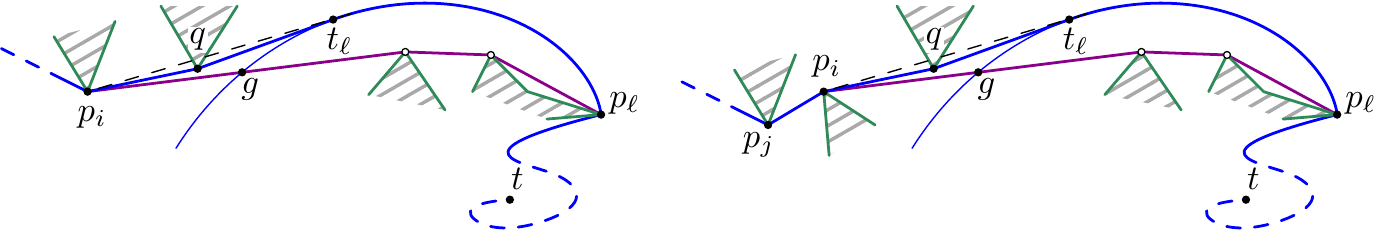}
\caption{A subpath of a shortest self-approaching $s$-$t$ path $\pi^{*}$ (in blue) between two consecutive inflection points of the geodesic path $\gamma$ (in purple) from $s$ to $t$ is geodesically convex. The last bend $q$ of $\pi^{*}$ before the vertex $p_{\ell}$ does not necessarily belong to $\gamma$.}
\label{fig:shortest-path-geodesic}%
\end{figure}

\begin{lem}\label{lem:geodesically-convex}
A shortest self-approaching $s$-$t$ path in a simple polygon $P$ consists of geodesically convex paths between inflection points of the geodesic from $s$ to $t$.
\end{lem}

\begin{theorem}
\label{thm:size}
A shortest self-approaching $s$-$t$ path in a simple polygon $P$ with $n$ vertices consists of $O(n^{2})$ segments. There exists a simple polygon $P$ and two points $s$ and $t$ in it, such that the shortest self-approaching from $s$ to $t$ has $\Omega(n^{2})$ segments.
\end{theorem}
\begin{proof}
To prove the upper bound, we will show that the number of segments on the convex hull of a shortest self-approaching path with $k$ bends is $O(k)$. Consider a bend $p_{\ell}$ of $\pi^{*}$, and let the convex hull of subpath $\pi^{*}(p_{\ell},t)$ have $T(\ell)$ segments. Let the previous bend before $p_{\ell}$ on $\pi^{*}$ be $p_{j}$, and denote the subpath of $\pi^{*}$ from $p_{j}$ to $p_{\ell}$ as $\sigma$. The subpath $\sigma$ consists of a straight segment, possibly, followed by several involute segments. Consider the evolute\footnote{If curve $\iota$ is an involute of curve $\sigma$, then $\sigma$ is called the evolute of $\iota$.} curves of these involutes. Except for maybe the last segment of the chain, each involute segment covers the complete range of the values of parameter $\theta$ of its evolute segment. Thus, this evolute lies completely inside the convex hull of the subpath $\pi^{*}(p_{j},t)$, and will no longer define other involutes. Notice, that this evolute can belong to the subpath $\sigma$ itself, if it winds around the convex hull of $\pi^{*}(p_{\ell}, t)$ multiple times. This fact does not change the size of the convex hull. Therefore, the convex hull of $\pi^{*}(p_{j})$ consists of no more than $T(\ell)+4$ segments: the number of involute segments grows by at most $1$, plus one new straight-line segment on $\sigma$, plus at most two straight-line segments of the common tangents to the convex hull of $\pi^{*}(p_{\ell},t)$ and the convex hull of $\sigma$. Setting the boundary condition $T(0)=2$, we conclude, that the number of segments on the convex hull of a shortest self-approaching path with $k$ bends is $O(k)$. The number of segments on a subpath $\sigma$ of $\pi^{*}$ between two bends is bounded by the size of the convex hull of the following path times the winding number of $\sigma$. And because $\pi^{*}$ has at most $O(n)$ bends, and it takes a constant number of vertices of $P$ to form a spiral for the path $\pi^{*}$ to wind around itself, its total size is $O(n^{2})$.

\begin{figure}[t]
\centering
\includegraphics{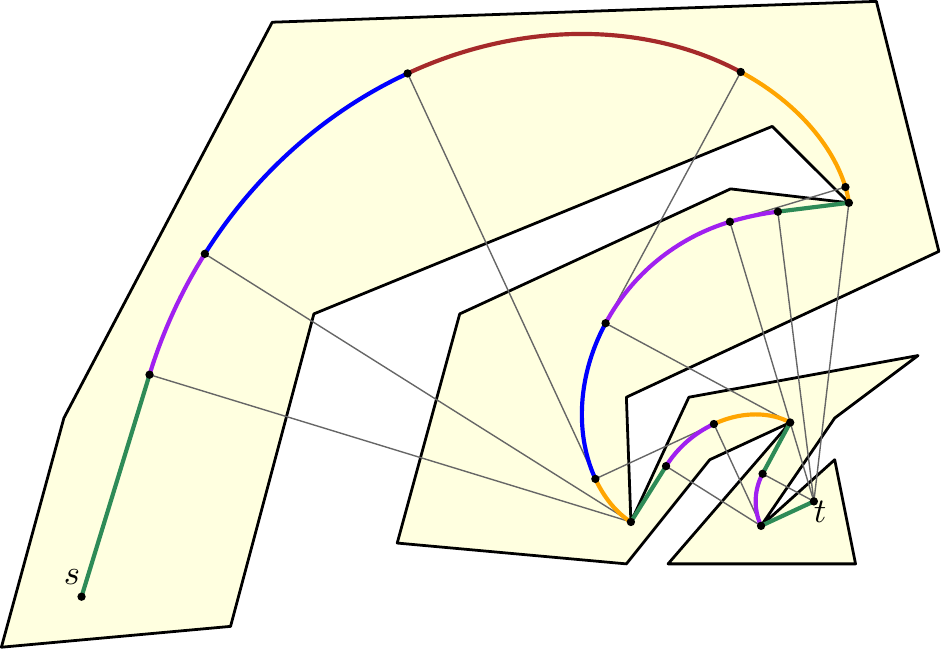}
\caption{Shortest self-approaching path from $s$ to $t$ consists of straight-line segments, circular arcs, and involutes of a circle of some order. Straight segments are shown in green, circular arcs in purple, involutes of a circle of first order in orange, involutes of a circle of a second order in blue, and involutes of a circle of third order in brown.}
\label{fig:shortest-path}%
\end{figure}

Consider a construction of a shortest self-approaching path in Figure~\ref{fig:shortest-path}. It shows how to linearly grow the number of segments on each subpath between the consecutive pairs of bends. Moreover, it is possible to force each bend with only a constant (specifically, four) number of the polygon boundary vertices per bend. Thus, for any $n$ there exists a polygon of size $n$ and two points $s$ and $t$ in it, such that the shortest self-approaching path between the points has $\Omega(n^{2})$ segments.
\end{proof}

\section{Existence of a self-approaching path}\label{sec:sa-path}

In this section we consider the question of testing whether, for given points $s$ and $t$ in a polygon $P$, they can be connected with a self-approaching path. In Theorem~\ref{thm:shortest-path} we proved that a shortest self approaching path can consist of involutes of a circle of a high order, and in Section~\ref{sec:prelim} we showed that such an involute is defined by a system of transcendental equations. In~\cite{Laczkovich2003} Laczkovich proved a strengthening of Richardson's theorem, which states that in general the statement $\exists x : f(x)=0$ is undecidable, where $f(x)$ is an expression generated by the rational numbers, the variable $x$, the operations of addition, multiplication, and composition, and the sine function. The Equations~(\ref{eq:involute}) that we need to solve to obtain formulas for the involutes are a special case of the class of expressions in Laczkovich's theorem. Nevertheless, it strongly suggests that an involute of a circle of order higher than one cannot be computed.

Next, we show an algorithm to test whether there exists a self-approaching path connecting two points $s$ and $t$, and if so, to compute the shortest path, under the assumption that we can solve Equations~(\ref{eq:involute}). 
Subsequently, it is possible to release this assumption, and modify the algorithm to build an approximate solution, given that the shortest self-approaching path from $s$-to-$t$ exists and there is a corridor of size $\varepsilon$ around it free of the polygon boundary points.

\subsection{Shortest path algorithm}

The proof of Theorem~\ref{thm:shortest-path} is constructive. Let us assume that we can solve equations of the form as Equations~(\ref{eq:involute}) for an involute of order $k$ in time $O(f(k))$, and evaluate the formula of the involute of order $k$ for a given parameter $\theta$ in time $O(g(k))$. Then, we can decide if two points $s$ and $t$ can be connected by a self-approaching path, and we can construct the shortest path between the points. The outline of the algorithm:
\begin{itemize}
\item Starting at $t$, move backwards along a geodesic $s$-$t$ path $\gamma$. Maintain the convex hull $\ch$ of the final part of the shortest self-approaching path $\pi^{*}$ to the destination $t$ built so far.
\item At every bend point $p_{\ell}$:
\begin{itemize}
\item Calculate the appropriate branch of an involute $I\!_{\ch}$ of the convex hull $\ch$. If $I\!_{\ch}$ intersects the opposite boundary of the polygon, thus, cutting off $s$ from $t$, report that a self-approaching path from $s$ to $t$ does not exist and terminate the algorithm.
\item Otherwise, find a geodesic path $\gamma_{\ell}$ from the preceding inflection point of $\gamma$ to $p_{\ell}$ in $P\backslash I\!_{\ch}$, and add its last segment $qp_{\ell}$ as a prefix to $\pi^{*}$.
\item Update the convex hull $\ch$. Repeat for the new bend point $q$, until $s$ is reached. Report the found path $\pi^{*}$.
\end{itemize}
\end{itemize}
To obtain an algorithm with an optimal running time, there are a few considerations to take into account when constructing the shortest path. First, instead of unnecessarily calculating the whole involute $I\!_{\ch}$ until the intersection point with the boundary of $P$, and then discarding the part of it under the tangent line from $q$, its segments can be calculated one by one as needed until the tangent point. Second, to optimally test if $I\!_{\ch}$ intersects the opposite boundary of the polygon, we can maintain a shortest path tree that will allow us to build funnels from the opposite sides of the polygon boundary. Third, it is not necessary to construct the whole geodesic $\gamma_{\ell}$ to be able to compute its last segment $qp_{\ell}$. Instead, we can move backwards along $\gamma$, vertex by vertex, until we reach a point from which a tangent to $I\!_{\ch}$ can be computed (possibly with adding new points along it).

Let the edges of $P$ be oriented in the counter-clockwise order. We shall call the two ends of an edge $e$, the \emph{front}-point, and the \emph{end}-point.

Next, we present the details of the algorithm.

\subparagraph{Initialization step.}
Compute a shortest path tree $\mathit{SPT}\!_{s}$ with the root $s$~\cite{Guibas1987}, and preprocess it to answer the lowest common ancestor query~\cite{Bender2000}. Compute the geodesic $\gamma$ from $s$ to $t$, and store $\gamma$ as a stack of vertices.
Let the first and the last segments of $\gamma$ be $\seg{sp'}$ and $\seg{p''t}$ respectively. Extend $\seg{sp'}$ beyond $s$ until intersection with $\partial P$ at some point $a$, and extend $\seg{p''t}$ beyond $t$ until intersection with $\partial P$ at some point $b$. This can be done in $O(\log n)$ time with a ray-shooting query after linear-time preprocessing of the polygon~\cite{Chazelle1994}. Let $L$ be the chain of the boundary of $P$ from $b$ to $a$ in the counter-clockwise order, we shall call it the \emph{left} chain. Similarly, let the \emph{right} chain $R$ be the chain of the boundary of $P$ from $a$ to $b$ in the counter-clockwise order.
Initialize $\pi^{*}$ and $\ch$ with the last segment $\seg{p''t}$ of $\gamma$, and pop the point $t$ from $\gamma$.

\subparagraph{The main loop.}
Let, before the beginning of the current iteration of the loop, $p_{\ell}$ be the point on top of the stack $\gamma$. Let $\pi^{*}$ touch $\partial P$ in point $p_{\ell}$ on its left side (the case when $\pi^{*}$ touches $\partial P$ on its right side is equivalent). Let $\ch$ be the convex hull of the already built subpath $\pi^{*}(p_{\ell},t)$.

Pop $p_{\ell}$ from the top of the stack $\gamma$. Consider the previous segment $\seg{p_{i}p_{\ell}}$ of $\gamma$ (point $p_{i}$ is currently on top of the stack $\gamma$).

\case{1} If the angle between $\seg{p_{\ell}p_{i}}$ and the tangent in the clockwise direction to $\ch$ at point $p_{\ell}$ form an angle that is not less than $90^{\circ}$, then $\seg{p_{i}p_{\ell}}$ lies on the shortest self-approaching path from $s$ to $t$; append $\pi^{*}(p_{\ell},t)$ with $\seg{p_{i}p_{\ell}}$ in the front, and update the convex hull. 

\case{2} If the angle between $\seg{p_{\ell}p_{i}}$ and the tangent in the clockwise direction to $\ch$ form an angle that is less than $90^{\circ}$, we need to calculate the involute $I_{\ch}$ of the convex hull for the tangent point moving clockwise around the boundary of $\ch$ starting at $p_{\ell}$. We first will determine until which point to calculate $I_{\ch}$.

\begin{figure}[t]
\centering
\includegraphics[page=1]{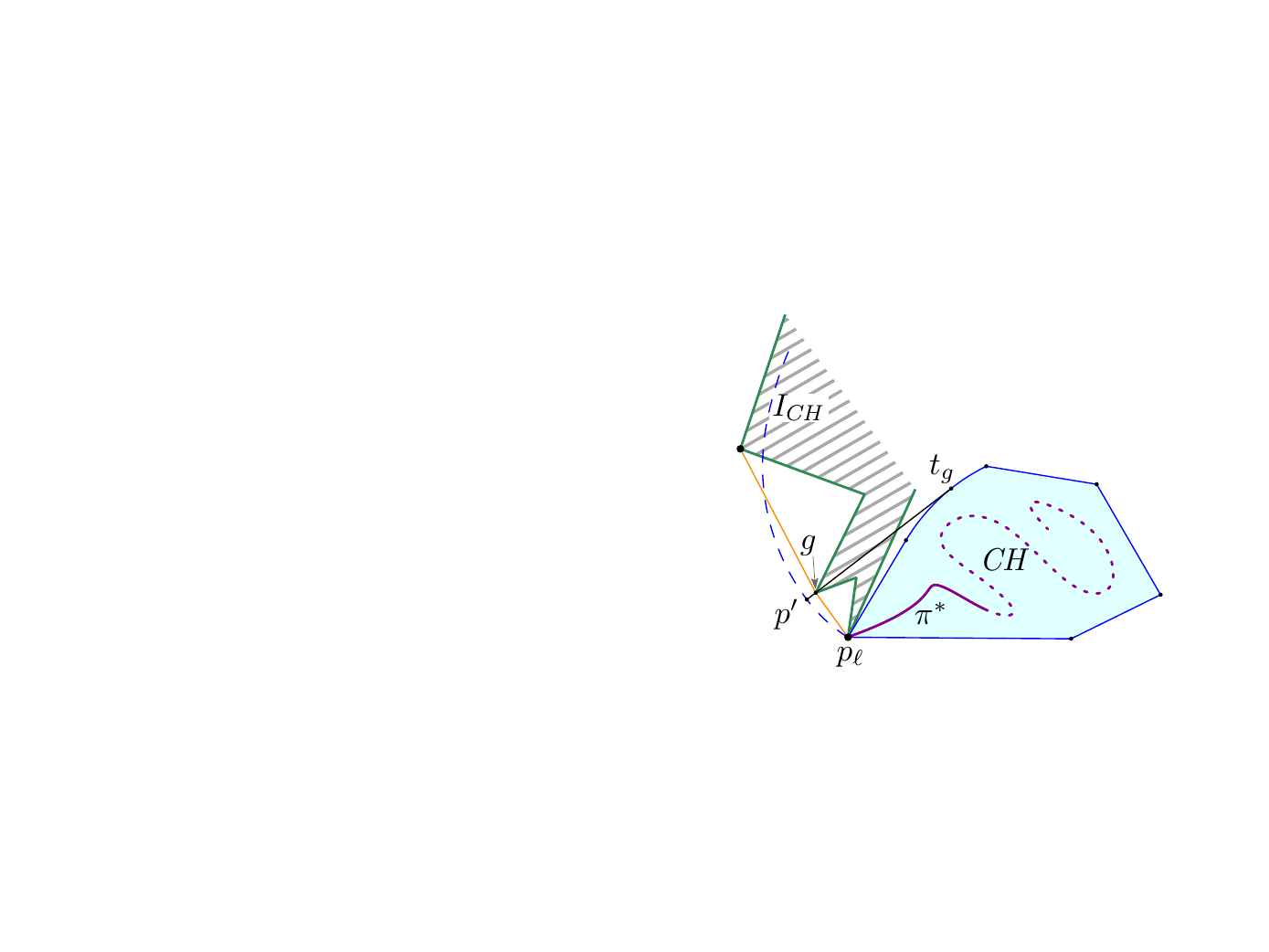}
\hfil
\includegraphics[page=3]{shortest-path-algo}
\caption{Illustration for Case 2 of the algorithm.}
\label{fig:algo1}%
\end{figure}

First, we check whether the points on the geodesic path $\gamma$ before $p_{\ell}$ lie in the dead region defined by $I_{\ch}$. To do that without explicitly constructing $I_{\ch}$ first, for each point $g$ on top of the stack $\gamma$, we construct a tangent line to  $\ch$ which is leaving it on its right side (refer to Figure~\ref{fig:algo1} (left)). Let $t_{g}$ be the tangent point on $\ch$, and let $t_{g}\in I_{g}$ for some involute segment $I_{g}$ on the boundary of the convex hull. Let $t_{g}g$ intersect $I_{\ch}$ at point $p'$. We know that the length of the segment $\seg{p't_{g}}$ is equal to the length of the boundary of the convex hull $\ch$ from $t_{g}$ to $p_{\ell}$. Thus, to check whether $g$ lies in the dead region we can compare the length of the segment $\seg{gt_{g}}$ with the length of the boundary of $\ch$ from $t_{g}$ to $g$. If $g$ does lie in the dead region, we simply remove it from the top of the stack $\gamma$, and proceed. If at some moment $\gamma$ becomes empty, \ie, the point $s$ lies in the dead region, we report that a self-approaching path from $s$ to $t$ does not exist and terminate the algorithm.

Now, let $p_{i}$ be the first point on $\gamma$ before $p_{\ell}$ that does not lie in the dead region of $I_{\ch}$. As in Figure~\ref{fig:shortest-path-geodesic}, the tangent segment from $p_{i}$ to the involute may intersect the right chain of the boundary of $P$. Moreover, the right chain of the boundary of $P$ may intersect $I_{\ch}$. To test and account for that case, we do the following. Let $\vec{\tau}=\vec{I}'_{\ch}(p_{\ell})$ be the tangent vector to $I_{\ch}$ at point $p_{\ell}$. Run a ray shooting query from $p_{\ell}$ in the direction $-\vec{\tau}$. Let it intersect an edge $e'$ of $R$, and denote its front-point as $p_{r}$ (refer to Figure~\ref{fig:algo1} (right)). Then, find a vertex $p_{j}$ in the shortest path tree $\mathit{SPT}\!_{s}$ that is the lowest common ancestor of $p_{\ell}$ and $p_{r}$. Let $\gamma_{\ell}$ and $\gamma_{r}$ be the two shortest paths from $p_{j}$ to $p_{\ell}$ and to $p_{r}$ respectively. Paths $\gamma_{\ell}$ and $\gamma_{r}$ form two convex chains. If $\gamma_{r}$ does not intersect $I_{\ch}$, then either a common tangent to $\gamma_{\ell}$ and $I_{\ch}$, or a common tangent to $\gamma_{r}$ and $I_{\ch}$, will belong to $\pi^{*}$. To be able to compute the common tangents, we now explicitly construct $I_{\ch}$ segment by segment until a certain point. Let $p'p''$ be the last segment of $I_{\ch}$ constructed so far (with the curve orientation from $p'$ to $p''$). We stop the construction of $I_{\ch}$ when the segment $\seg{p'p_{j}}$ makes a left turn with respect to the tangent vector $-\vec\tau$, where $\vec\tau=\vec{I}'_{\ch}(p')$.

Whether $\gamma_{r}$ intersects $I_{\ch}$ can be found during the computation of the common tangent. If it does, report that $s$ and $t$ cannot be connected with a self-approaching path and terminate the algorithm.

Let $q_{\ell}$ and $q_{r}$ be the two tangent points on $\gamma_{\ell}$ and $\gamma_{r}$ respectively of the common tangent lines with $I_{\ch}$. One of the points $q_{\ell}$ and $q_{r}$, or both, will be equal to $p_{j}$.
\begin{itemize}
\item If $p_{j}=q_{\ell}=q_{r}$, then append $\pi^{*}$ with $\seg{p_{j}t_{j}}\oplus I_{\ch}(t_{j},p_{\ell})$, where $t_{j}$ is the tangent point on $I_{\ch}$.
\item If $p_{j}=q_{r}\not=q_{\ell}$, then append $\pi^{*}$ with $\gamma(p_{j},q_{\ell})\oplus \seg{q_{\ell}t_{\ell}}\oplus I_{\ch}(t_{\ell},p_{\ell})$, where $t_{\ell}$ is the tangent point on $I_{\ch}$ of the common tangent with $\gamma_{\ell}$.
\item If $p_{j}=q_{\ell}\not=q_{r}$, then append $\pi^{*}$ with $\gamma(p_{j},q_{r})\oplus \seg{q_{r}t_{r}}\oplus I_{\ch}(t_{r},p_{\ell})$, where $t_{r}$ is the tangent point on $I_{\ch}$ of the common tangent with $\gamma_{r}$. 
\end{itemize}
Remove the points from $\gamma$ until $p_{j}$ is on top of the stack, and update $\ch$. Iterate over the main loop until $\gamma$ is empty, and return $\pi^{*}$.

\subparagraph{Computing common tangents.} During the execution of the algorithm, we need to be able to compute common tangents between a convex polygonal chain and a convex chain of involutes, and between two convex chains of involutes. It is possible to compute a common tangent of two non-intersecting simple polygons of size $m$ and $n$ in $O(\log (m+n))$ time~\cite{Kirkpatrick1995}. We will use this as the first step in finding a common tangent to our chains.

\begin{figure}[t]
\centering
\includegraphics[page=1]{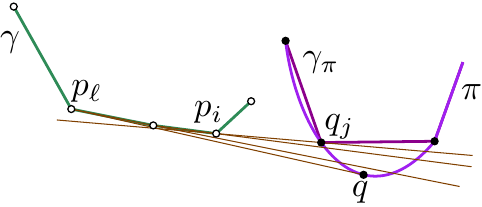}\hfil\includegraphics[page=2]{bitangent}
\caption{Common tangent to a polygonal chain and a chain of involutes.}
\label{fig:bitangent}%
\end{figure}

Consider a convex polygonal chain $\gamma$ of size $m$ and a convex chain of involutes $\pi$ of size $n$, with points ordered in the counter-clockwise order. First, we will show how to compute an outer common tangent. Without loss of generality, we will show how to compute a tangent $pq$, where $p\in\gamma$, $q\in\pi$, and $\gamma$ and $\pi$ lie on the left of $\vec{pq}$.  Denote $\gamma_\pi$ to be a polygonal chain connecting the end points of segments of $\pi$ in the consistent order. Let the convex hulls of $\gamma$ and $\gamma_\pi$ be disjoint. Build a common tangent $\seg{p_{i} q_{j}}$ between $\ch(\gamma)$ and $\ch(\gamma_\pi)$ in $O(\log(m+n))$ time (refer to Figure~\ref{fig:bitangent}), where $p_{i}\in\gamma$ and $q_{i}\in\gamma_{\pi}$. The common tangent between $\gamma$ and $\pi$ does not necessarily go through the point $p_{i}$, but it does touch one of the adjacent segment to the point $q_{j}$. Let, without loss of generality, the common tangent touch the segment $q_{j}q_{j+1}$ of $\pi$. Then, using a binary search on the chain $\gamma$ we can find a point $p_{\ell}$ which will belong to the common tangent. For that point $p_{\ell}$, the ray $p_{\ell}p_{\ell+1}$ intersects $q_{j}q_{j+1}$, and the ray $p_{\ell-1}p_{\ell}$ does not. To test whether a line $\ell$ intersects a given involute, we can compute the tangent line to this involute parallel to $\ell$, and check in which half-plane of $\ell$ the tangent point belongs to. Computing the tangent point is equivalent to evaluating the formula of the involute of $q_{j}q_{j+1}$ for the parameter $\theta$ which corresponds to the direction perpendicular to $\ell$. Thus, the tangent point $q$ can be found in $O(g(k)\log m)$ time, where $k$ is the order of the involute. After that, we check if the ray intersects the involute by checking the turn angle of the segment of $\gamma$ defining the ray, and the segment from its end point to the tangent point $q$. Consider ray $p_{\ell}p_{\ell+1}$. If $p_{\ell}p_{\ell+1}q$ is making a right turn under $90^{\circ}$, the ray $p_{\ell}p_{\ell+1}$ intersects the involute segment $q_{j}p_{j+1}$. The last step is to calculate the tangent segment $\seg{pq}$ from $p$ to $q_{j}q_{j+1}$ in $O(f(k))$ time. Therefore, a common tangent to a polygonal chain and a chain of involutes can be constructed in $O(\log(m+n)+g(k)\log m + f(k))$ time.

To compute an inner common tangent of a convex polygonal chain $\gamma$ of size $m$ and a convex chain of involutes $\pi$ of size $n$, we follow the same steps as in the case of the outer common tangent, except for the rays from the segments of $\gamma$ which will be emanating in the opposite direction.

A small modification will allow us to construct a common tangent between two convex chains of involutes $\pi_{1}$ and $\pi_{2}$. Start with computing a common tangent $\tau$ for the two polygonal chains $\gamma_{\pi_{1}}$ and $\gamma_{\pi_{2}}$ connecting the end points of segments of $\pi_{1}$ and $\pi_{2}$ respectively. Let $\tau$ touch $\gamma_{\pi_{1}}$ and $\gamma_{\pi_{2}}$ at points $p_{i}$ and $q_{j}$. If $\tau$ does not intersect any of the four adjacent involute segments to $p_{i}$ and $q_{j}$, then we have found the common tangent segment $\seg{p_{i}q_{j}}$. If $\tau$ intersects two involute segments adjacent to $p_{1}$ and $p_{2}$, then the common tangent line to $\pi_{1}$ and $\pi_{2}$ will touch these involute segments. Let, without loss of generality, $\tau$ intersect two involute segments $p_{i}p_{i+1}$ and $q_{j}q_{j+1}$. The common tangent line to these two segments will define the parameter $\theta$ which, in its turn, will determine the tangent points $p$ and $q$ on the two involute segments. Note, that it follows that the orders of involutes of $p_{i}p_{i+1}$ and $q_{j}q_{j+1}$ must both be even, or must both be odd. The common tangent line can be found by solving one of the following equations (depending on the parity of the order of the involutes):
\[
I_{1}(\theta)=I_{2}(\theta)+c\begin{pmatrix}\cos\theta\\ \sin\theta\end{pmatrix}\qquad\text{or}\qquad I_{1}(\theta)=I_{2}(\theta)+c\begin{pmatrix}-\sin\theta\\ \cos\theta\end{pmatrix}\,.
\]
This system of equations is similar to the one of Equations~\ref{eq:involute}. We find the tangent segment $\seg{pq}$ in $O(f(k))$ time, where $k$ is the maximum order of the two involutes. Finally, if $\tau$ intersects only one involute segment, we repeat similar steps to finding the tangent between the polygonal chain and the chain of involutes. Overall, it takes the same asymptotic running time $O(\log(m+n)+g(k)\log m + f(k))$ to compute the tangent segment.

If the two chains $\gamma$ and $\pi$ may be intersecting, we start with checking if $\gamma$ and $\gamma_{\pi}$ are intersecting in $O(\log(m+n))$ time~\cite{Chazelle1987}. Afterwards, during the binary search, for each point of $\gamma$ considered, we check in $O(f(k))$ time if the point is inside the involute. If no points inside the involute were found, and the binary search returned a candidate point $p_{\ell}$, we check if segment $\seg{p_{\ell-1}p_{\ell}}$ intersects the involute (with both points $p_{\ell-1}$ and $p_{\ell}$ being outside). We do this again by finding the tangent point and checking which half-plane it lies in. Overall, for two chains $\gamma$ and $\pi$, we can test if they intersect, and if not, find the inner tangent segment, in $O(\log(m+n)+(g(k)+f(k))\log m)$ time.

\subparagraph{Maintaining $\ch$.} At the end of each iteration of the main algorithm, we need to update the convex hull of the subpath of the shortest self-approaching path built so far. This can involve finding a tangent from a point to a chain of involutes, or finding a common tangent of two chains of involutes.

Moreover, we want to be able to optimally calculate the length of a boundary from the current point $p_{\ell}$ to some point $t_{g}$. For that, associate two values $\dist_{cw}(u)$ and $\dist_{ccw}(u)$ with each end point of a segment on $\ch$ that will contain the distance to $p_{\ell}$ (up to some constant that will be equal for all the points) along the boundary in the clockwise and counter-clockwise direction respectively. Moreover, for two points $u$ and $v$ on $\ch$, the length of the boundary between them can be calculated by $\dist_{ccw}(v)-\dist_{ccw}(u)$, if the chain of $\ch$ between $u$ and $v$ in the counter-clockwise order does not contain $p_{\ell}$. This fact will allow us to maintain the values in the points unchanged when updating the convex hull.

At every iteration of the algorithm, the distance from some tangent point $t_{g}$ on an involute segment $p'p''$ to $p_{\ell}$ in the clockwise direction can be computed by formula $s(t_{g})=\length_{I}(t_{g},p'')+\dist_{cw}(p'')-\dist_{cw}(p_{\ell})$, where $\length_{I}(t_{g},p'')$ is the arc length of the involute from point $t_{g}$ to $p''$. Analogously, the distance from $t_{g}$ to $p_{\ell}$ in the counter-clockwise direction can be computed by taking $s(t_{g})=\length_{I}(t_{g},p')+\dist_{ccw}(p')-\dist_{ccw}(p_{\ell})$.

When updating the convex hull after extending the path $\pi^{*}$, we calculate the lengths of the tangent segments and the new involute arcs, and set the values $\dist_{cw}(u)$ and $\dist_{ccw}(u)$ to the new points of $\ch$ relatively to the values of the points remaining on $\ch$. This will take $f(k)$ time to compute the arc length per segment of an involute of order $k$.

Taking these considerations into account, we conclude with the following theorem:
\begin{theorem}
\label{thm:main-algo}
The algorithm above constructs a shortest self-approaching path from $s$ to $t$ or reports that it does not exist in $O(K+\frac{n\log K}{\sqrt{K}}(g(\sqrt{K})+f(\sqrt{K})))$ running time, where $K$ is the size of the output.
\end{theorem}
\begin{proof}
At every iteration of the algorithm, new segments are added to $\pi^{*}$, and thus, by Theorem~\ref{thm:size}, the algorithm will terminate.

First, we prove that the path $\pi^{*}$ that the algorithm produces is self-approaching. In Case 1, a straight-line segment is added to $\pi^{*}$. A normal to any point not his segment does not intersect the convex hull of the future path, and thus does not intersect the future path itself. In Case 2, an involute segment and one or a number of straight-line segments are added to $\pi^{*}$. Consider a point on the involute part. By construction, a normal to the involute at the point is tangent to the convex hull of the future path, and thus does not intersect it. The rest of the path added to $\pi^{*}$ is either a straight-line segment tangent to the involute, or a convex polygonal chain which is bending in the opposite direction as the involute segment and has the total turn angle not more than $90^{\circ}$. In both cases, the normals to any point on these segments do not intersect the future path. Thus, if the algorithm outputs a path, it is self-approaching.

Moreover, by construction, the path does not have convex bends, and only has a positive curvature at on the boundaries of dead regions. In other words, the path is geodesic in $P\backslash\cup D$. Any shorter path will intersect one of the dead regions, and thus will not be self-approaching.

The algorithm reports that a self-approaching path does not exist only in two cases: (1) if $s$ lies in a dead region; (2) if an involute to a convex hull separates $s$ from $t$. In both cases, any path from $s$ to $t$ will have to cross some dead region, and thus cannot be self-approaching.

The running time of the initialization step of the algorithm is $O(n)$.

The bottleneck in the running time of one iteration of the algorithm for Case 1 is the time it takes to update the convex hull. In the case when only one segment is being added to $\pi^{*}$, the update consists of computing two tangents from a point to the convex hull. This can be done in $O(\log|\ch|+f(k))$ time, where $k$ is the order of the involute to which the range point belongs. After that, the values $\dist_{cw}$ and $\dist_{ccw}$ are updated in constant time.

In Case 2, we spend $O(f(k))$ time for each vertex of the geodesic path on the left polygonal boundary $L$ that lies in the dead region. One ray-shooting query takes $O(\log n)$ time. Building an outer tangent between the chain from $p_{i}$ to $p_{j}$ and the involute takes $O(\log(m+|CH|)+g(k)\log m + f(k))$ time, where $m$ is the length of the chain, and the number of segments on the convex hull of the involute is bounded by the size of the convex hull \ch. Note, that the length of the involute chain may be more than $|CH|$ if it winds around the convex hull \ch several times. However, we only need its outer layer to be able to construct the tangent line. Building an inner tangent between the chain from $p_{r}$ to $p_{j}$ and the involute, including testing for intersections, takes $O(\log(m+|CH|)+(g(k)+f(k))\log m)$ time, where $m$ is the length of the chain. Updating the convex hull includes, possibly, computing a tangent from a point to $\ch$, or a common tangent between $\ch$ and a chain of involutes. This can take up to $O(f(k))$ time. After that, the values $\dist_{cw}$ and $\dist_{ccw}$ are updated in constant time per segment, that is in total in $O(|CH|+m)$ time. If $K$ is the size of the output, \ie, the number of segments of $\pi^{*}$, and $m$ is about $\frac{n}{\sqrt{K}}$ to maximize the value of $\sum\log m_{i}$, the total running time adds up to $O(K+\frac{n\log K}{\sqrt{K}}(g(\sqrt{K})+f(\sqrt{K})))$.
\end{proof}

\section{Self-approaching polygon}\label{sec:sa-poly}

A polygon is self-approaching, if for any two points there exists a self-approaching path connecting them.

\begin{theorem}
\label{thm:sa-poly}
Polygon $P$ is self-approaching if and only if for any disk $D$ centered at any point $p\in P$, the intersection $D\cap P$ has one connected component.
\end{theorem}
\begin{proof}
Let polygon $P$ be self-approaching. We will show that for any disk $D$ centered at any point $p\in P$, the intersection $D\cap P$ has one connected component. Assume that there exist a point $p$ and a disk $D$ centered at $p$ such that intersection $D\cap P$ has more than one connected component. Then choose any point $s$ inside a connected component other than the one containing $p$. Then any $s$-to-$p$ path will cross the boundary of the disk, and therefore such path will not be self-approaching.

\begin{figure}[t]
\centering
\includegraphics{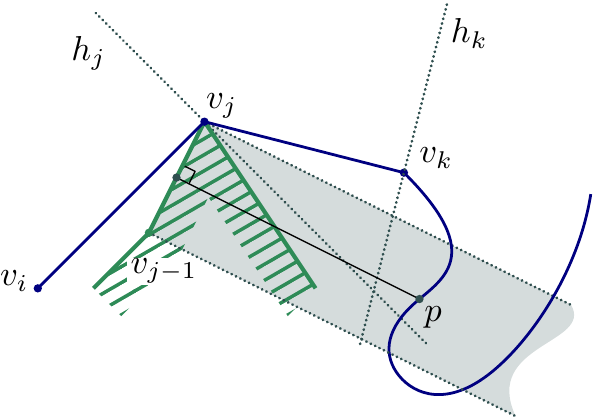}%
\caption{Illustration to Theorem~\ref{thm:sa-poly}: if $D\cap P$ has one connected component for any disk $D$ centered at any point $p\in P$, a geodesic path between any two points $s$ and $t$ is self-approaching.}
\label{fig:sa-poly}%
\end{figure}

Let the intersection $D\cap P$ have one connected component for any disk $D$ centered at any point $p\in P$. We will show that for any two points $s$ and $t$ in $P$ there exists an $s$-$t$ self-approaching path inside $P$. Consider a disk $D_t$ centered at $t$ with radius $|ts|$. $D_t\cap P$ has one connected component, therefore a geodesic $\gamma$ from $s$ to $t$ lies completely inside $D_t\cap P$. We will show that $\gamma$ is self-approaching. Suppose that $\gamma$ is not self-approaching, then consider a segment $\seg{v_i v_j}$ of $\gamma$ closest to $t$ such that a perpendicular $h_j$ to $\seg{v_i v_j}$ drawn through $v_j$ intersects subpath $\gamma(v_j,t)$. Let $v_k$ be the next vertex of $\gamma$ after $v_j$, and let $h_k$ be a perpendicular to $\seg{v_j v_k}$ drawn through $v_k$. Without loss of generality, let $\{v_i, v_j, v_k\}$ make a right turn (refer to Figure~\ref{fig:sa-poly}). By assumption $\seg{v_i v_j}$ is the last segment of $\gamma$ such that $h_j$ intersects $\gamma(v_j,t)$, therefore $h_k$ does not intersect $\gamma(v_k,t)$, and therefore $h_j$ intersects $\gamma(v_j,t)$ to the right from $\seg{v_i v_j}$. Consider segment $\seg{v_{j-1} v_j}$ on the boundary of polygon $P$ such that $\seg{v_{j-1} v_j}$ is facing $\seg{v_i v_j}$. Then a strip between two perpendiculars to $\seg{v_{j-1} v_j}$ drawn through $v_{j-1}$ and $v_j$ intersects $\gamma(v_j,t)$. Select any point $p$ on $\gamma(v_j,t)$ inside the strip. There exists a disk $D_p$ centered at $p$ that intersects the segment $\seg{v_{j-1} v_j}$ twice, and therefore $D_p\cap P$ has more than one connected component. Thus, a geodesic between any two points in $P$ is self-approaching.
\end{proof}

Recall, that a path is increasing chord if it is self-approaching in both directions.
\begin{corollary}
Any self-approaching polygon is also increasing-chord.
\end{corollary}

Next, we present an algorithm to test whether a given simple polygon $P$ is self-approaching. Observe, that from the proof of Theorem~\ref{thm:sa-poly} the following property holds: the polygon $P$ is self-approaching if and only if, for all edges $e$ on the boundary of $P$ directed in the counter-clockwise order, an area bounded between the two normals to $e$ at its two end points in the right half-plane of $e$ is free of $\partial P$. We call this area the \emph{half-strip} of $e$. We will use this property to test efficiently if the polygon is self-approaching.

Let $P$ be given as a set of points $p_{0},p_{1},\dots,p_{n-1}$ in the counter-clockwise order around the boundary. We will start at $p_{0}$, move along the boundary in the counter-clockwise order and maintain the union of all the half-strips of the edges visited so far. More precisely, we will maintain the left and the right sides, $\rho_{l}$ and $\rho_r$, of the hour-glass shape that is the union of the half-strips; $\rho_l$ and $\rho_r$ are convex polygonal chains (refer to Figure~\ref{fig:sa-poly-2}). Store the segments of $\rho_l$ and $\rho_r$ as two lists, the last segments in the lists are infinite rays.

\begin{figure}[t]
\begin{minipage}[t]{0.5\textwidth}
\qquad\quad
\textbf{\textsf{(a)}}
\qquad
\includegraphics[page=1]{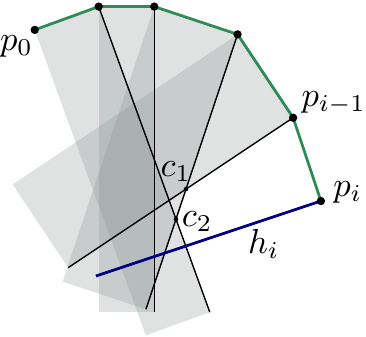}%
\end{minipage}%
\begin{minipage}[t]{0.5\textwidth}
\qquad\quad
\textbf{\textsf{(b)}}
\qquad
\includegraphics[page=3]{sa-poly-2}%
\end{minipage}\\
\begin{minipage}[t]{0.5\textwidth}
\qquad\quad
\textbf{\textsf{(c)}}
\qquad
\includegraphics[page=2]{sa-poly-2}%
\end{minipage}%
\begin{minipage}[t]{0.5\textwidth}
\qquad\quad
\textbf{\textsf{(d)}}
\qquad
\includegraphics[page=4]{sa-poly-2}%
\end{minipage}
\caption{The illustration for Theorem~\ref{thm:sa-poly-alg}.}
\label{fig:sa-poly-2}%
\end{figure}

At every iteration of the algorithm, perform the following steps. Let $p_{i}$ be the current point of the polygon $P$. The chain $\rho_r$ contains the right side of the union of all the half-strips up to point $p_{i-1}$. Consider the next boundary segment $\seg{p_{i-1}p_{i}}$, and a perpendicular ray $h_{i}$ at the point $p_{i}$ (refer to Figure~\ref{fig:sa-poly-2} (a)). To update the chain $\rho_r$, do the following: Traverse $\rho_r$, and for every its segment $\seg{c_{j}c_{j+1}}$,
\begin{itemize}
\item if $\seg{p_{i-1}p_{i}}$ intersects $\seg{c_{j}c_{j+1}}$, then report that $P$ is not self-approaching and terminate;
\item if $h_{i}$ intersects $\seg{c_{j}c_{j+1}}$, calculate the intersection point $c'$, and replace the first elements of the list $\rho_r$ up to $\seg{c_{j}c_{j+1}}$ with two segments, $\seg{p_{i}c'}$ and $\seg{c'c_{j+1}}$; repeat for the next point $p_{i+1}$.
\end{itemize}
Traverse the boundary of polygon $P$ twice in the counter-clockwise order, and then repeat the same algorithm traversing the boundary of $P$ twice in the clockwise order. If none of the segments $\seg{p_{i-1}p_{i}}$ intersected a segment of $\rho_r$, report that $P$ is self-approaching.

\begin{theorem}\label{thm:sa-poly-alg}
Given a simple polygon $P$ with $n$ vertices, the presented algorithm tests in $O(n)$ time if it is self-approaching.
\end{theorem}
\begin{proof}
Consider the counter-clockwise traversal of the boundary of $P$. There are two cases when the boundary segment $\seg{p_{i-1}p_{i}}$ intersects $\rho_r$. In the first case, $\seg{p_{i-1}p_{i}}$ intersects $\rho_r$, and $h_{i}$ does not intersect it (refer to Figure~\ref{fig:sa-poly-2} (b)). Let us call it the intersection of type $1$. In the second case, $\seg{p_{i-1}p_{i}}$ intersects $\rho_r$ after $h_{i}$ intersects it (refer to Figure~\ref{fig:sa-poly-2} (c)). Let us call it the intersection of type $2$. Moreover, the boundary segment $\seg{p_{i-1}p_{i}}$ may intersect $\rho_l$ (refer to Figure~\ref{fig:sa-poly-2} (d)). Let us call it the intersection of type $3$.

When traversing the polygon counter-clockwise, the presented algorithm will recognize the first type of the intersection, but not the second or third type. In case of the second type, during one iteration, the algorithm stops traversing $\rho_r$ after finding the intersection point of $h_{i}$, and thus will not find the intersection of the segment with $\rho_{r}$. And in case of the third type, the algorithm does not check for intersection with $\rho_{l}$ at all.

Nevertheless, we will prove, that by repeating the checks above twice and in two directions, counter-clockwise from $p_{0}$ to $p_{n-1}$, and clockwise from $p_{n-1}$ to $p_{0}$, the algorithm will correctly decide if the polygon is self-approaching or not.

\case{0} If the polygon is self-approaching, then none of the segments will intersect $\rho_r$ or $\rho_l$. The algorithm will traverse the polygon twice, then twice in the clockwise direction, and report that it is self-approaching.

\case{1} If only the first intersection type occurs, then the algorithm will traverse the boundary of $P$ until the first violation of the half-strip property, correctly report that the polygon is not self-approaching, and terminate.

\case{2} Suppose that the second intersection type occurs. Consider the first segment $\seg{p_{i-1}p_{i}}$, such that both $h_{i}$ and $\seg{p_{i-1}p_{i}}$ intersect $\rho_r$. Let $\seg{p_{i-1}p_{i}}$ intersect the normal to some preceding segment $\seg{p_{j-1}p_{j}}$ at the point $p_{j}$. As the ray $h_{i}$ intersects $\rho_r$ before $\seg{p_{i-1}p_{i}}$ does, it also intersects the polygon boundary between the points $p_{j}$ and $p_{i-1}$. And, therefore, the ray $h_{i-1}$ perpendicular to $\seg{p_{i-1}p_{i}}$ at the point $p_{i-1}$ also intersects the polygon boundary between the points $p_{j}$ and $p_{i-1}$. Then, consider the behavior of the algorithm during the backwards traversal. Let $p_{\ell}$ for $j\le\ell<i-1$ be the first point on the left side of the ray $h_{i-1}$. Then the segment $\seg{p_{\ell+1},p_{\ell}}$ intersects $h_{i-1}$, and either the segment $\seg{p_{\ell+1},p_{\ell}}$ intersects the left chain $\rho_{l}$ or there was another segment before $\seg{p_{\ell+1},p_{\ell}}$ that intersected $\rho_{l}$. Note, that because the intersection of $\seg{p_{i-1}p_{i}}$ and $\rho_r$ was the first violation of the half-strip property in the counter-clockwise order, the intersection of $\seg{p_{\ell+1},p_{\ell}}$ and $\rho_{l}$ cannot be of the second type, otherwise $p_{i-1}$ would already lie on the right side of a normal to $\seg{p_{\ell+1},p_{\ell}}$ at the point $\seg{p_{\ell+1}}$. Therefore, this intersection can only be of type one, and the algorithm will recognize it during the backwards traversal.

\case{3} Suppose that the third intersection type occurs. Then, there will be a segment $\seg{p_{j+1},p_{j}}$ (where $j\ge i$), for which the first or the second intersection type occurs when traversing the polygon in the opposite direction, and thus either case 1 or case 2 applies.

Thus, we only need to explicitly check for the first intersection type. The running time of the algorithm is $O(n)$. At every iteration, the number of segments removed from the list $\rho_{r}$ is equal to half the number of tests for intersections the algorithm makes, and the number of segments added back is at most $2$. Therefore, the total number of segments that can be removed from $\rho_{r}$ over one traversal of the boundary is not more than $2n$. Similarly, the total number of segments that can be removed from $\rho_{l}$ over one traversal of the boundary  is not more than $2n$. Therefore, the algorithm performs $O(n)$ intersection tests.
 \end{proof}

\section{Reachable and reverse-reachable regions}\label{sec:rev-reach-regs}

\begin{figure}[t]
\centering
\includegraphics{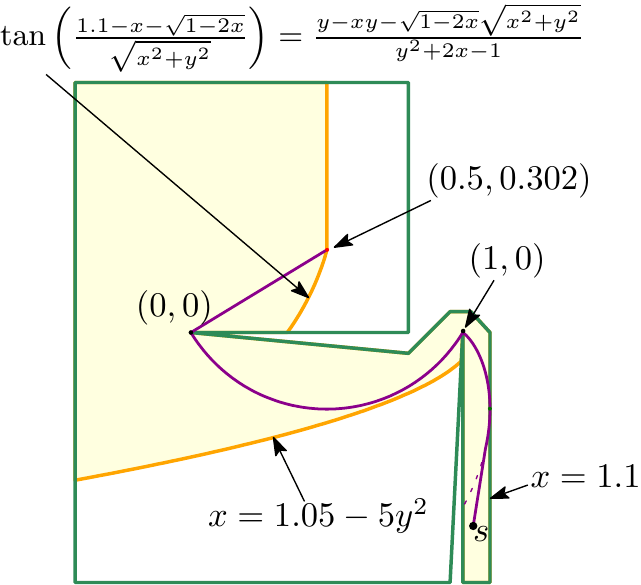}
\caption{An instance of a polygon $P$ and a point $s$ such that the boundary of the reachable region $\R(s)$ contains transcendental curves.}
\label{fig:reachable-example}%
\end{figure}

Recall, that the set of points that can be reached from $s$ by a self-approaching path in $P$ is the reachable region of $s$; and the set of points from which point $t$ can be reached with a self-approaching path in $P$ is the reverse-reachable region of $t$.

\subparagraph{Reachable regions.} A region, reachable from a point in a simple polygon, can have a complicated structure, and have transcendental equations defining its boundary. Figure~\ref{fig:reachable-example} shows an example of a reachable region, whose boundary is described by a transcendental equation already after the second turn of a self-approaching path.

Nevertheless, reachable regions seem to have some nice properties. We conjecture that,
\begin{conjecture}
Reachable region $\R(s)$ is geodesically convex. That is, for any two points $t_{1}$ and $t_{2}$ that are reachable from $s$ by two self-approaching paths, any point on the shortest path connecting $t_{1}$ and $t_{2}$ is also reachable from $s$ by a self-approaching path.
\end{conjecture}

\subparagraph{Reverse-reachable regions.}
Reverse-reachable regions $\RR(t)$ (or their approximations) can be constructed by a modified algorithm for finding a shortest self-approaching path for two given points. We start at the destination point $t$ and traverse the shortest path tree $\mathit{SPT}(t)$ in a breadth-first search manner, building a self-approaching shortest path tree. For every iteration, we construct a dead region by building an involute of the convex hull of the current leaf-to-root self-approaching path, and check for the intersections of the dead regions with the boundary of $P$ that may cut off parts of the polygon. The remaining region $P\backslash\cup D$ is the reverse-reachable region. As for the algorithm of computing the shortest self-approaching path, it takes $O(n^{2}+f(n)\log n)$ time to construct the reverse-reachable region. Then, a shortest $s$-$t$ path query for a query point $s$ can be answered by finding a tangent from $s$ to the boundary of $\RR(t)$, and then following the appropriate branch of the shortest path tree.

\section*{Acknowledgements}

This work was begun at the CMO-BIRS Workshop on Searching and Routing in Discrete and Continuous Domains, October 11--16, 2015. Irina Kostitsyna was supported in part by the NWO under project no. 612.001.106, and by F.R.S.-FNRS.

\bibliographystyle{plain}
\bibliography{sa-paths-arxiv}

\end{document}